\documentclass[letterpaper, 10 pt, conference]{ieeeconf}
\IEEEoverridecommandlockouts
\def\BibTeX{{\rm B\kern-.05em{\sc i\kern-.025em b}\kern-.08em
    T\kern-.1667em\lower.7ex\hbox{E}\kern-.125emX}}

\usepackage{xcolor}
\usepackage{graphicx}
\usepackage{float}
\usepackage{amssymb}
\usepackage{amsmath}
\usepackage{amsthm}
\newtheorem{thm}{Theorem}
\newtheorem{defn}{Definition}

\newtheorem{lem}{Lemma}
\newtheorem{prob}{Problem}

\newtheorem{assm}{Assumption}
\newtheorem{reform}{Reformulation}

\newcommand{\R}{\mathbb{R}}
\newcommand{\N}{\mathbb{N}}
\renewcommand{\P}{\mathbb{P}}
\newcommand{\E}{\mathbb{E}}

\newcommand{\pr}[1]{\mathbb{P}\!\left(#1\right)}
\newcommand{\ex}[1]{\mathbb{E}\!\left[#1\right]}
\newcommand{\hex}[1]{\hat{\mathbb{E}}\!\left[#1\right]}

\newcommand{\var}[1]{\mathrm{Var}\!\left(#1\right)}
\newcommand{\hvar}[1]{\hat{\mathrm{Var}}\!\left(#1\right)}
\newcommand{\std}[1]{\mathrm{Std}\!\left(#1\right)}
\newcommand{\hstd}[1]{\hat{\mathrm{Std}}\!\left(#1\right)}

\newcommand{\bvec}[1]{\vec{\boldsymbol{#1}}}
\newcommand{\Nt}[2]{\mathbb{N}_{[#1,#2]}}

\usepackage[nodisplayskipstretch]{setspace}
\makeatletter
\let\NAT@parse\undefined
\makeatother
\usepackage{hyperref}
\hypersetup{
    colorlinks=true,      
    linkcolor=black,
    citecolor=black,
    filecolor=black,
    urlcolor=black     
    }

\author{Shawn Priore and Meeko Oishi
    \thanks{
        This material is based upon work supported by the National Science Foundation under NSF Grant Number CMMI-2105631. Any opinions, findings, and conclusions or recommendations expressed in this material are those of the authors and do not necessarily reflect the views of the National Science Foundation.  
        \newline \indent Shawn Priore and Meeko Oishi are with Elec. \& Computer Eng., Univ. of New Mexico, Albq., NM; e-mail: \texttt{shawn.a.priore@gmail.com} (corresponding author) and \texttt{oishi@unm.edu}.
    }
}
\title{Stochastic Optimal Control For Gaussian Disturbances with Unknown Mean and Variance Based on Sample Statistics}

\begin{document}

\maketitle

\begin{abstract}
    We propose an open loop methodology based on sample statistics to solve chance constrained stochastic optimal control problems with probabilistic safety guarantees for linear systems where the additive Gaussian noise has unknown mean and covariance. We consider a joint chance constraint for time-varying polytopic target sets under assumptions that the disturbance has been sufficiently sampled. We derive two theorems that allow us to bound the probability of the state being more than some number of sample standard deviations away from the sample mean. We use these theorems to reformulate the chance constraint into a series of convex and linear constraints. Here, solutions guarantee chance constraint satisfaction. We demonstrate our method on a satellite rendezvous maneuver and provide comparisons with the scenario approach.  
\end{abstract}

\section{Introduction}

Stochastic systems with incomplete disturbance information can arise in many safety critical stochastic systems. For example, autonomous cars can experience disturbances from poor road conditions caused by weather or poor maintenance, varying road materials changing drag coefficients, or gusty winds attempting to blow the vehicle off the road. One method to handle incomplete disturbance information is to use data-driven control, which has proven to be useful for systems in which traditional modeling techniques cannot be used. 
However, data-driven methods that rely on sample data typically cannot guarantee constraint satisfaction as they only provide an approximation of the true disturbance.

The use of data-driven approaches for control has been the topic of extensive literature. The scenario approach \cite{calafiore2006scenario, Campi2008} solves the joint chance constraints by finding a controller that satisfies each of the constraints for all samples in a set of disturbance samples. The scenario approach can be used to derive a confidence bound for the probability of the joint chance constraints being satisfied, but lacks strict assurances of satisfaction \cite{Yang2019}. While the scenario approach allows for a simple interpretation and easy implementation, the main drawback is that the number of samples required is a function of both the probabilistic violation threshold and the confidence bound. When either the probabilistic violation thresholds or the confidence bounds are very strict, the number of samples required can be prohibitively large \cite{KariotoglouSCL2016}. Iterative approaches \cite{Campi2018TAC, care2014fast} been proposed to reduce the computational burden by reducing the number of samples at each iteration and comparing successive solutions. 

Data reduction techniques relying on parameter estimation have been posed to simplify sample-based approaches \cite{Saha2010, Madankan15, RAJAMANI09}. These methods, however, tend to incorporate these estimates as ground truth \cite{Verma10} and can lead to maneuvers that are not safe when implemented. Similarly, robust control mechanisms \cite{Ben-Tal2009, Lam15} have also been used to find solutions by creating bounds based on the extremum of the sample data. Robust bounds can over- or under-approximate the probability of the chance constraint based on how representative the sample data is of the true distribution. Both the parameter estimation techniques and robust control mechanisms that rely on data cannot enforce strict chance constraint satisfaction as they approximate the distribution without regard for the distribution of these estimates.

In this paper we derive a concentration inequality that bounds the tail probability of a random variable based on sample statistics under unimodality conditions. We propose the application of this concentration inequality for the evaluation of polytopic target set chance constraints in the presence of Gaussian disturbances with unknown mean and covariance. 
The main contribution of this paper is a {\em closed-form} reformulation of polytopic target set chance constraints with sample statistics that {\em guarantees chance constraint satisfaction}. Although the derived concentration inequality cannot be applied as broadly as other data-driven methods, it does allow for guarantees of chance constraint satisfaction for LTI systems with Gaussian distributions with unknown mean and covariance. The main drawback of our approach is the derived concentration inequality is known to be conservative for evaluation of Gaussian chance constraints.


The paper is organized as follows. Section \ref{sec:prelim} provides mathematical preliminaries and formulates the problem. Section \ref{sec:methods} derives the sample-based concentration inequalities and reformulates the joint chance constraint into convex constraints. Section \ref{sec:results} demonstrates our approach on satellite rendezvous problems, and we conclude in Section \ref{sec:conclusion}.

\section{Preliminary and Problem Setup} \label{sec:prelim}

We denote random variables with bold, for vectors $\bvec{x} \in \R^n$ and scalars $\boldsymbol{x} \in \R$. For a random variable $\boldsymbol{x}$, we denote the expectation as $\ex{\boldsymbol{x}}$, standard deviation as $\std{\boldsymbol{x}}$, and variance as $\var{\boldsymbol{x}}$. The $i$\textsuperscript{th} sample of $\boldsymbol{x}$ is $\boldsymbol{x}^{[i]}$. Estimates are denoted with $\hat{\cdot}$, i.e., $\hex{\boldsymbol{x}}$ is the sample mean of $\boldsymbol{x}$. We denote the Kronecker product of matrices $A$ and $B$ as $A \otimes B$. The $n \times n$ identity matrix is $I_n$, a $n \times m$ matrix of zeros is $0_{n \times m}$ or $0_{n}$ if $n=m$, and a $n\times 1$ vector of ones is $\vec{1}_n$.

\subsection{Problem Formulation}

Consider a planning context in which the evolution of a vehicle is governed by the discrete-time LTI system,
\begin{equation}
    \bvec{x}(k+1) = A \bvec{x}(k) + B \vec{u}(k) + \bvec{w}(k) \label{eq:system}
\end{equation}
with state $\bvec{x}(k) \in \mathcal{X} \subseteq \R^n$, input $\vec{u}(k) \in \mathcal{U} \subseteq \R^m$,  $\bvec{w}(k) \in \R^n$, and known initial condition $\vec{x}(0)$.  We presume the admissible control set, $\mathcal{U}$, is a convex polytope, and that the system evolves over a finite time horizon of $N \in \N$ steps. We presume the disturbance $\bvec{w}(k)$ follows a Gaussian distribution 
\begin{equation}
    \bvec{w}(k) \sim \mathcal{N}\left( \ex{\bvec{w}(k)}, \var{\bvec{w}(k)} \right)
\end{equation}
with unknown mean $\ex{\bvec{w}(k)}$ and unknown variance-covariance matrix $\var{\bvec{w}(k)}$. 

We write the dynamics at time $k$ 
in concatenated form as
\begin{equation} \label{eq:lin_dynamics}
    \bvec{x}(k) = A^k \vec{x}(0) + \mathcal{C}(k) \vec{U} + \mathcal{D}(k) \bvec{W}
\end{equation}
with 
$\vec{U} =\left[ \vec{u}(0)^\top \ \ldots \ \vec{u}(N-1)^\top \right]^\top \in \mathcal{U}^{N}$, 
$\bvec{W} = \left[ \bvec{w}(0)^\top \ \ldots \ \bvec{w}(N-1)^\top \right]^\top \in \mathbb{R}^{Nn}$,
$\mathcal{C}(k) =  \left[ A^{k-1}B \ \ldots \ AB \ B \ 0_{n \times (N-k)m} \right] \in \R^{n \times Nm}$, and 
$\mathcal{D}(k) = \left[ A^{k-1} \ \ldots \ A \ I_n \ 0_{n \times (N-k)n} \right] \in \R^{n \times Nn}$.

We presume the vehicle has a desired time-varying polytopic target set. This state constraint is considered probabilistically and must hold with desired likelihood. Formally, 
\begin{equation}
    \pr{\cap_{k=1}^{N}\bvec{x}(k)  \in  \mathcal{T}(k)}  \geq  1-\alpha \label{eq:constraint_t}
\end{equation}
We presume convex, compact, and polytopic sets $ \mathcal{T}(k) \subseteq \mathbb R^n$,  and probabilistic violation threshold $\alpha < 1/6$, which is required as a condition for optimality of our method.  

We seek to solve the following optimization problem, with convex performance objective $J: \mathcal{X}^{N} \times \mathcal{U}^{N} \rightarrow \R$,
\begin{subequations}\label{eq:prob_1_opt}
    \begin{align}
        \underset{\vec{U}}{\mathrm{minimize}} \quad & J\left(
        \bvec{X}, \vec{U}\right)  \\
        \mathrm{subject\ to} \quad  & \vec{U} \in \mathcal{U}^N,  \\
        & \text{Dynamics } \eqref{eq:lin_dynamics} \text{ with }
        \vec{x}(0)\\
        & \text{Probabilistic constraint \eqref{eq:constraint_t}} \label{eq:prob_1_opt_constraints} 
    \end{align}
\end{subequations}
where $\bvec{X}$ is the concatenated state vector. 

\begin{defn}[Almost Surely \cite{casella2002}]
    Let $(\Omega, \mathcal{B}(\Omega), \P)$ be a probability space with outcomes $\Omega$, Borel $\sigma$-algebra $\mathcal{B}(\Omega)$, and probability measure $\P$. An event $\mathcal{A} \in \mathcal{B}(\Omega)$ happens almost surely if $\pr{\mathcal{A}} = 1$ or $\pr{\mathcal{A}^{c}}=0$ where $\cdot^{c}$ denotes the complement of the event.
\end{defn}

We introduce several key assumptions to make \eqref{eq:prob_1_opt} a tractable optimization problem. 

\begin{assm} \label{assm:samples}
The concatenated disturbance vector $\bvec{W}$ has been independently sampled $N_s$ times. We denote the sampled values as $\bvec{W}^{[i]}$ for $i \in \Nt{1}{N_s}$
\end{assm}
\begin{assm} \label{assm:n_samples}
The sample size $N_s$ must be sufficiently large such that the reformulations presented in this work are tractable.
\end{assm}
\begin{assm} \label{assm:samples_not_equal}
The concatenated disturbance vector samples $\bvec{W}^{[i]}$ are almost surely not all equal. 
\end{assm}

Assumptions \ref{assm:samples}-\ref{assm:n_samples} are required to compute sample statistics. Assumption \ref{assm:n_samples} guarantees that the theorems developed here can be applied for our reformulations. 
Assumption \ref{assm:samples_not_equal} guarantees that the sample standard deviation is greater than zero, implying the distribution is not degenerate nor deterministic. 


\begin{prob} \label{prob:1}
    Under Assumptions \ref{assm:samples}-\ref{assm:samples_not_equal}, solve the stochastic optimization problem \eqref{eq:prob_1_opt} with probabilistic violation threshold $\alpha$ for an open loop controller $\vec{U} \in  \mathcal{U}^N$.
\end{prob} 

The main challenge in solving Problem \ref{prob:1} is assuring \eqref{eq:prob_1_opt_constraints}. Historically, methods that rely on sample data, such as the scenario approach, cannot guarantee the derived controller satisfy \eqref{eq:prob_1_opt_constraints} \cite{Yang2019}. Without knowledge of the underlying mean and covariance structure, analytic techniques cannot be used to derive reformulations that allow for guarantees. 

\section{Methods} \label{sec:methods}

We solve Problem \ref{prob:1} by reformulating each chance constraint as an affine summation of the random variable's sample mean and sample standard deviation. This form is amenable to the concentration inequality derived in Section \ref{ssec:bounds}, allowing for a convex reformulation with probabilistic guarantees.

\subsection{Establishing A Sample Based Concentration Inequality} \label{ssec:bounds}

Here, we state the the pivotal theorem that allow us to solve Problem \ref{prob:1}. For brevity, the proof is in Appendix \ref{appx:out-sample}. 

\begin{thm} \label{thm:out_sample_vp}
Let $\boldsymbol{x}$ follow some distribution $f$. Let $\boldsymbol{x}^{[1]},\dots, \boldsymbol{x}^{[N_s]}$ be samples drawn independently from the distribution $f$, for some $N_s\geq2$. Let
\begin{subequations} \label{eq:stats}
\begin{align}
    \hex{\boldsymbol{x}} = & \; \frac{1}{N_s} \sum_{i=1}^{N_s} \boldsymbol{x}^{[i]} \\
    \hstd{\boldsymbol{x}} =  & \; \sqrt{\frac{1}{N_s} \sum_{i=1}^{N_s} (\boldsymbol{x}^{[i]} - \hex{\boldsymbol{x}})^2}
\end{align}
\end{subequations}
be the sample mean and sample standard deviation, respectively, with $\hstd{\boldsymbol{x}}>0$ almost surely. Then, if the distribution of the statistic
\begin{equation}
    \frac{\boldsymbol{x}^{[i]}-\hex{\boldsymbol{x}}}{ \hstd{\boldsymbol{x}}}
\end{equation}
is unimodal,
\begin{equation} \label{eq:out_sample_vp}
    \pr{\boldsymbol{x}\!-\!\hex{\boldsymbol{x}} \!\geq\! \lambda \hstd{\boldsymbol{x}}} \!\leq\! \frac{4(\sqrt{N_s\!+\!1} + \lambda)^2}{9 \left(\lambda^2 N_s \!+\! (\sqrt{N_s\!+\!1}\!+\!\lambda)^2\right) }
\end{equation}
for any 
\begin{equation} \label{eq:thm_lambda_restrict}
    \lambda>\frac{\sqrt{5(N_s+1)}}{\sqrt{3N_s}-\sqrt{5}}
\end{equation}
\end{thm}

Theorem \ref{thm:out_sample_vp} provides a bound for deviations of a random variable $\boldsymbol{x}$ from the sample mean of the distribution. For the purpose of generating open loop controllers in an optimization framework, this will allow us to bound chance constraints based on sample statistics of previously collection data.

For brevity, we define $N_s^{\ast} = N_S + 1$ and 
\begin{equation} \label{eq:lambda_func}
    f(\lambda) = \frac{4(\sqrt{N_s^{\ast}} + \lambda)^2}{9 \left(\lambda^2 N_s + (\sqrt{N_s^{\ast}} + \lambda)^2\right)} 
\end{equation}

To address the need for Assumption \ref{assm:n_samples}, we observe that
\begin{equation} \label{eq:f_lim}
    \lim_{\lambda \rightarrow \infty} f(\lambda) = \frac{4}{9 N_s^{\ast}}
\end{equation}
For any probabilistic violation threshold smaller than this value, Theorem \ref{thm:out_sample_vp} will not be sufficiently tight to bound the constraint. Figure \ref{fig:ineq} graphs \eqref{eq:lambda_func} for $N_s$ taking the values 10, 100, 1000, and $\infty$.

\begin{figure}
    \centering
    \includegraphics[width=0.8\columnwidth]{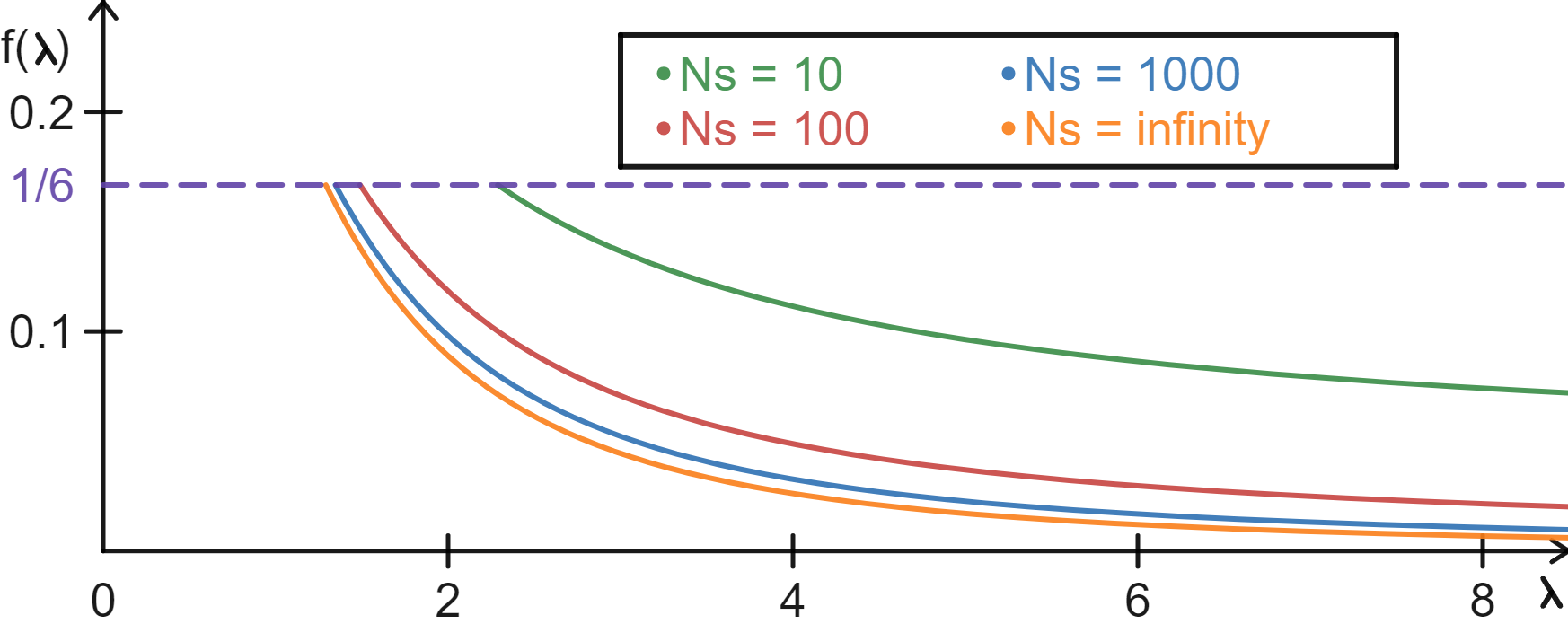}
    \caption{Graph of \eqref{eq:lambda_func} for values of $N_s \in \{10,100,1000,\infty\}$ with the restriction \eqref{eq:thm_lambda_restrict}. }
    \label{fig:ineq}
\end{figure}

Note that we did not use Bessel's correction \cite{casella2002} in the sample variance formula to simplify the probabilistic bound. Accordingly, the sample variance statistic we used is biased in relation to the variance of the distribution. An analogous result can be derived with Bessel's correction, however, the bound becomes more complex.

\subsection{Polytopic Target Set Constraint} \label{ssec:target_reform}

Next, we reformulate \eqref{eq:constraint_t}. First, we write the polytope $\mathcal{T}(k)$ as the intersection of $N_{T_k}$ half-space inequalities,
\begin{equation} \label{eq:polytope}
    \mathcal{T}(k) =   \bigcap_{i=1}^{N_{T_k}} \vec{G}_{ik} \bvec{x}(k) \leq h_{ik}
\end{equation}
where $\vec{G}_{ik} \in \R^n$ and $h_{ik} \in \R$. From \eqref{eq:polytope}, we write \eqref{eq:constraint_t} as
\begin{equation}
    \pr{\bigcap_{k=1}^{N}\boldsymbol{x}(k) \!\in\! \mathcal{T}(k)}\! = \! \pr{ \bigcap_{k=1}^{N}\!\bigcap_{i=1}^{N_{T_k}} \vec{G}_{ik} \bvec{x}(k) \leq h_{ik}}
\end{equation}
By taking the complement and employing Boole's inequality, we separate the combined chance constraints into a series of individual chance constraints,
\begin{equation}
\pr{\bigcup_{k=1}^{N}\boldsymbol{x}(k) \!\not\in\! \mathcal{T}(k)}
\leq  \sum_{k=1}^{N} \sum_{i=1}^{N_{T_k}} \pr{ \vec{G}_{ik} \bvec{x}(k) \geq h_{ik}}
\end{equation}

\noindent Using the approach in \cite{ono2008iterative}, we introduce risk allocation variables $\omega_{ik}$ for each of the individual chance constraints,
\begin{subequations}\label{eq:quantile_reform_new_var}
\begin{align}
     \pr{  \vec{G}_{ik} \bvec{x}(k) \geq h_{ik}} &\leq \omega_{ik} \label{eq:quantile_orig} \\
     {\textstyle \sum_{k=1}^{N}\sum_{i=1}^{N_{T_k}} }  \omega_{ik} &\leq \alpha \label{eq:quantile_reform_new_var_2}\\
     \omega_{ik} & \geq 0 \label{eq:quantile_reform_new_var_3}
\end{align}
\end{subequations}
As the mean and variance of $\bvec{w}(k)$ are unknown, the mean and variance of the random variable $\vec{G}_{ik} \bvec{x}(k)$ are also unknown. Hence, we cannot evaluate \eqref{eq:quantile_orig} directly. 

To facilitate reformulation of \eqref{eq:quantile_reform_new_var}, we introduce 
sample statistics for samples collected per Assumption \ref{assm:samples}.
\begin{subequations} \label{eq:stats_for_w}
\begin{align}
    \hex{\bvec{W}} = & \; \frac{1}{N_s} \sum_{[i]=1}^{N_s} \bvec{W}^{[i]} \\
    \hvar{\bvec{W}} = & \;  \frac{1}{N_s}  \sum_{[i]=1}^{N_s} \left(\bvec{W}^{[i]} \!-\!\hex{\bvec{W}}\right)\left(\bvec{W}^{[i]} \!-\!\hex{\bvec{W}}\right)^{\top}
\end{align}
\end{subequations}
Here, $\hex{\bvec{W}}$ is the sample mean of the disturbance vector $\bvec{W}$ and $\hvar{\bvec{W}}$ is the sample variance-covariance matrix of the disturbance vector $\bvec{W}$. Then each constraints that make up the polytopic target set constraint will have sample mean
\begin{equation} \label{eq:sample_poly_mean}
    \hex{\vec{G}_{ik} \bvec{x}(k)}
     = \vec{G}_{ik}\!  \left(A^k \vec{x}(0) \!+\! \mathcal{C}(k) \vec{U} \!+\! \mathcal{D}(k) \hex{\bvec{W}} \right)
\end{equation} 
and sample standard deviation
\begin{equation} \label{eq:sample_poly_var}
    \hstd{\vec{G}_{ik} \bvec{x}(k)}
     = \sqrt{\vec{G}_{ik} \mathcal{D}(k) \hvar{\bvec{W}} \mathcal{D}^{\top}(k)\vec{G}_{ik}^{\top}}
\end{equation}
for $i \in \Nt{1}{N_{T_k}}$. 

Our goal is to use Theorem \ref{thm:out_sample_vp} to reformulate \eqref{eq:quantile_reform_new_var} into a tractable form. Hence we must show the following items are true,
\begin{enumerate}
    \item The statistic
        \begin{equation} \label{eq:unimodal_constraint}
            \frac{\vec{G}_{ik} \bvec{x}(k)^{[i]} - \hex{\vec{G}_{ik} \bvec{x}(k)}}{\hstd{\vec{G}_{ik} \bvec{x}(k)}}
        \end{equation}
        elicits a unimodal distribution.
    \item Adding the constraint 
        \begin{equation}
            \hex{\vec{G}_{ik} \bvec{x}(k)} + \lambda_{ik} \hstd{\vec{G}_{ik} \bvec{x}(k)} \leq h_{ik}
        \end{equation} 
        and the restriction \eqref{eq:thm_lambda_restrict} allows us to solve \eqref{eq:quantile_reform_new_var} via the variable substitution $\omega_{ik} = f(\lambda_{ik})$. 
\end{enumerate}
We will show these in order.
For brevity we denote
\begin{equation}
     \hat{\mathbb{F}}(\bvec{x}(k), \lambda_{ik}) = \hex{\vec{G}_{ik} \bvec{x}(k)} + \lambda_{ik} \hstd{\vec{G}_{ik} \bvec{x}(k)}
\end{equation}
First, we establish the unimodality of the statistic \eqref{eq:unimodal_constraint}.

\begin{lem} \label{lem:unimodal}
The distribution of \eqref{eq:unimodal_constraint} is unimodal for $N_s \geq 4$.
\end{lem}
\begin{proof}
By the properties of the Gaussian distribution, we can observe that
\begin{equation}
     \vec{G}_{ik} \bvec{x}(k)^{[i]} \sim \mathcal{N} \left( \vec{G}_{ik} \ex{ \bvec{x}(k)}, \vec{G}_{ik} \var{ \bvec{x}(k)} \vec{G}_{ik}^{\top}\right)
\end{equation}
for all $i \in \Nt{1}{N_s}$. Then the random variable
\begin{equation} 
    \boldsymbol{z}^{[i]} = \frac{\vec{G}_{ik} \bvec{x}(k)^{[i]} - \hex{\vec{G}_{ik} \bvec{x}(k)}}{\hstd{\vec{G}_{ik} \bvec{x}(k)}}
\end{equation}
elicits the Lebesgue probability density function (pdf) \cite[eq. 3.1]{Shao2003}
\begin{equation}
\begin{split}
    & \phi_{\boldsymbol{z}^{[i]}}(z) = \\
    & \ \frac{\Gamma(\frac{N_s-1}{2})}{\sqrt{\pi (N_s\!-\!1)} \Gamma(\frac{N_s\!-\!2}{2})} \left[ 1\!-\! \frac{ z^2}{N_s\!-\!1}\right]^{\frac{N_s\!-\!4}{2}} I_{|z| \leq \sqrt{N_s\!-\!1}} 
\end{split}
\end{equation}
for all $[i] \in \Nt{1}{N_s}$ where $\Gamma (\cdot)$ is the gamma function and $I_{cond}$ is an indicator function taking value 1 when the condition $cond$ is true. Observe that the second derivative of the log of the pdf is
\begin{equation} \label{eq:second_deriv}
    \frac{\partial^2}{\partial z^2} \log (\phi_{\boldsymbol{z}^{[i]}}(z))  = -\frac{\left(N_s-1+z^2\right)\left(N_s-4\right)}{\left(N_s-1-z^2\right)^2}
\end{equation}
for $|z| <\sqrt{N_s\!-\!1}$. Here, the sign of \eqref{eq:second_deriv} is determined by the term $N_s-4$. Hence, the pdf is log concave when $N_s \geq 4$. Since log concave distributions are also unimodal \cite{Bertin1997}, we can conclude that \eqref{eq:unimodal_constraint} is unimodal when $N_s \geq 4$. 
\end{proof}

To use Theorem \ref{thm:out_sample_vp}, we must impose the restriction $N_s \geq 4$ per Lemma \ref{lem:unimodal}. This restriction will be in addition to Assumption \ref{assm:n_samples}.

Now we address the second point. We add the constraints
\begin{subequations} \label{eq:new_const}
    \begin{align}
        \hat{\mathbb{F}}(\bvec{x}(k), \lambda_{ik}) & \leq h_{ik} \label{eq:add_target}\\
        \lambda_{ik} & >\frac{\sqrt{5 N_s^{\ast}}}{\sqrt{3N_s}-\sqrt{5}}
    \end{align}
\end{subequations}
to \eqref{eq:quantile_reform_new_var} with risk allocation variables $\lambda_{ik}$. Observe that if \eqref{eq:add_target} holds, then
\begin{equation} \label{eq:enforce_target}
\begin{split}
     \pr{\vec{G}_{ik} \bvec{x}(k) \!\geq\! h_{ik}}  \!\leq\! \pr{\vec{G}_{ik} \bvec{x}(k) \!\geq\! \hat{\mathbb{F}}(\bvec{x}(k), \lambda_{ik})}
\end{split}
\end{equation} 
for any $\lambda_{ik} > 0$. Here, we know that $\hstd{\vec{G}_{ik} \bvec{x}(k) } >0$ almost surely by Assumption \ref{assm:samples_not_equal}. Then by Assumptions \ref{assm:samples}-\ref{assm:samples_not_equal}, the restriction $N_s \geq 4$, Lemma \ref{lem:unimodal}, and Theorem \ref{thm:out_sample_vp}, 
\begin{equation} \label{eq:upperbound}
    \pr{\vec{G}_{ik} \bvec{x}(k) \geq \hat{\mathbb{F}}(\bvec{x}(k), \lambda_{ik})} \leq f(\lambda_{ik})
\end{equation}
Combining \eqref{eq:quantile_reform_new_var} and \eqref{eq:new_const}-\eqref{eq:upperbound}, we get 
\begin{subequations} \label{eq:combined}
\begin{align}
    \pr{  \vec{G}_{ik} \bvec{x}(k) \geq h_{ik}} &\leq \omega_{ik}  \label{eq:combined_1}\\
    \pr{\vec{G}_{ik} \bvec{x}(k) \!\geq\! \hat{\mathbb{F}}(\bvec{x}(k), \lambda_{ik})}  & \geq \pr{  \vec{G}_{ik} \bvec{x}(k) \!\geq\! h_{ik}} \\
    \pr{\vec{G}_{ik} \bvec{x}(k) \!\geq\! \hat{\mathbb{F}}(\bvec{x}(k), \lambda_{ik})} & \leq f(\lambda_{ik}) \label{eq:combined_3}\\
    \hat{\mathbb{F}}(\bvec{x}(k), \lambda_{ik})  & \leq h_{ik} \\
    {\textstyle \sum_{k=1}^{N}\sum_{i=1}^{N_{T_k}} }  \omega_{ik} &\leq \alpha \\
    \omega_{ik} & \geq 0 \label{eq:combined_6}\\
    \lambda_{ik} & >\frac{\sqrt{5 N_s^{\ast}}}{\sqrt{3N_s}-\sqrt{5}}  \label{eq:combined_7}
\end{align}
\end{subequations}
By using the variable substitution $\omega_{ik} = f(\lambda_{ik})$, \eqref{eq:combined} simplifies to

\begin{subequations} \label{eq:combined_small}
\begin{align}
    \pr{  \vec{G}_{ik} \bvec{x}(k) \geq h_{ik}}  & \leq f(\lambda_{ik}) \label{eq:combined_small_1}\\
    \hat{\mathbb{F}}(\bvec{x}(k), \lambda_{ik}) & \leq h_{ik} \label{eq:combined_small_2} \\
    {\textstyle \sum_{k=1}^{N}\sum_{i=1}^{N_{T_k}} }  f(\lambda_{ik}) &\leq \alpha \label{eq:combined_small_3} \\
    \lambda_{ik} & > {\textstyle \frac{\sqrt{5 N_s^{\ast}}}{\sqrt{3N_s}-\sqrt{5}}} \label{eq:combined_small_4}
\end{align}
\end{subequations}
Here, \eqref{eq:combined_1}-\eqref{eq:combined_3} combine into \eqref{eq:combined_small_1}, and \eqref{eq:combined_6} is no longer required as $f(\lambda_{ik}) >0$ by the restriction \eqref{eq:combined_small_4}. Further, since $\alpha < 1/6$, \eqref{eq:combined_small_3} implies \eqref{eq:combined_small_4}, making it a redundant constraint. Similarly, \eqref{eq:combined_small_1} acts as an intermediary condition between \eqref{eq:combined_small_2} and \eqref{eq:combined_small_3}, making it a redundant constraint. Hence, we can further simplify \eqref{eq:combined_small} to 
\begin{subequations}\label{eq:target_constraint}
\begin{align}
    \hat{\mathbb{F}}(\bvec{x}(k), \lambda_{ik}) \leq & \; h_{ik}  \label{eq:target_reform}\\ 
    {\textstyle \sum_{k=1}^{N}\sum_{i=1}^{N_{T_k}} }  f(\lambda_{ik}) \leq & \; \alpha  \label{eq:target_lambda}
\end{align}
\end{subequations}
with \eqref{eq:target_reform} iterated over the index $i \in \Nt{1}{N_{T_k}}$.

\begin{lem} \label{lem:target_satisfy}
For the controller $\vec{U}$, if there exists risk allocation variables $\lambda_{ik}$ satisfying \eqref{eq:target_constraint} for constraints in the form of \eqref{eq:constraint_t}, then $\vec{U}$ satisfies \eqref{eq:prob_1_opt_constraints} almost surely.
\end{lem}

\begin{proof}
By construction, \eqref{eq:sample_poly_mean}-\eqref{eq:sample_poly_var} is equivalent to \eqref{eq:stats} for the random variable $\vec{G}_{ik} \bvec{x}(k)$.  By Assumption \ref{assm:samples_not_equal}, $\hstd{\vec{G}_{ik} \bvec{x}(k) } >0$. Here, $N_s \geq 4$ is stricter than $N_s \geq 2$. The statistic \eqref{eq:unimodal_constraint} is unimodal by Lemma \ref{lem:unimodal}. Satisfaction of \eqref{eq:target_lambda} implies \eqref{eq:thm_lambda_restrict} holds. Hence, all the conditions for Theorem \ref{thm:out_sample_vp} have been met. Theorem \ref{thm:out_sample_vp}, Boole's inequality, and De Morgan's law \cite{casella2002} guarantee that \eqref{eq:prob_1_opt_constraints} is satisfied when \eqref{eq:target_constraint} is satisfied.
\end{proof}

\begin{lem}
The constraint reformulation \eqref{eq:target_constraint} will always be convex in $\vec{U}$ and $\lambda_{ik}$. 
\end{lem}

\begin{proof}
By construction \eqref{eq:sample_poly_mean} is affine in the control input and \eqref{eq:sample_poly_var} is constant with respect to the input. By the affine constructions of \eqref{eq:target_reform} and \eqref{eq:sample_poly_mean}, \eqref{eq:target_reform} is affine and hence convex, in $\vec{U}$ and $\lambda_{ik}$. 

Observe the second partial derivative of $f(\lambda)$ with respect to $\lambda$ is
\begin{equation}
    \frac{\partial^2}{\partial \lambda^2}f(\lambda) = \frac{8N_s\left(\lambda^3(N_s^{\ast})^{3/2}+3\lambda^2(N_s^{\ast})^2-(N_s^{\ast})^2\right)}{9\left(N_s\lambda^2+\left(\lambda+\sqrt{N_s^{\ast}}\right)^2\right)^3}    
\end{equation}
Then $f(\lambda)$ has inflection points where 
\begin{equation} \label{eq:cubic_eq}
    \frac{2}{\sqrt{N_s^{\ast}}}\lambda^3+3\lambda^2 -1 = 0
\end{equation}
The function \eqref{eq:cubic_eq} has three real roots with the only positive root being \cite{Zucker2008}
\begin{equation} \label{eq:cubic_root}
    \lambda = \underbrace{\sqrt{N_s^{\ast}}\left[\cos\left(\frac{1}{3}\arccos\left(-\frac{N_s-1}{N_s^{\ast}}\right)\right)-\frac{1}{2}\right]}_{\Theta(N_s)}
\end{equation} 
Further, $\lambda > \Theta(N_s) \Leftrightarrow f''(\lambda)>0$ implying that $f(\lambda)$ is convex in this region. Note that 
\begin{equation} \label{eq:lower_bound_lambda}
    \frac{\sqrt{5N_s^{\ast}}}{\sqrt{3N_s}-\sqrt{5}} \geq \Theta(N_s)
\end{equation}
hold for all values of $N_s$. This implies that $f(\lambda)$ is always convex under the restriction \eqref{eq:target_lambda} as $\alpha<1/6$. Since, \eqref{eq:target_lambda} is the sum of convex functions, it too is convex. Finally, in the problem formulation, we defined the control authority to be a closed and convex set. Therefore, the chance constraint reformulation \eqref{eq:target_constraint} will always be convex.
\end{proof}

We formally write out the final form of the constraint \eqref{eq:target_constraint} by combining \eqref{eq:target_constraint} with the sample mean and sample standard deviation formulas, \eqref{eq:sample_poly_mean} and \eqref{eq:sample_poly_var}, respectively. We do so in \eqref{eq:target_constraint_3}. 
\begin{figure*}
\begin{subequations}\label{eq:target_constraint_3}
\begin{align}
    \vec{G}_{ik}\left( A^k \vec{x}(0) + \mathcal{C}(k) \vec{U} + \mathcal{D}(k) \hex{\bvec{W}} \right) + \lambda_{ik} \sqrt{\vec{G}_{ik}\mathcal{D}(k) \hvar{\bvec{W}} \mathcal{D}^{\top}(k)\vec{G}_{ik}{\top}} \leq & \; h_{ik} & \forall i \in \Nt{1}{N_{T_k}} \label{eq:target_reform_3}\\ 
    {\textstyle \sum_{k=1}^{N} \sum_{i=1}^{N_{T_k}}} f(\lambda_{ik}) \leq & \; \alpha & \label{eq:target_lambda_3}
\end{align}
\end{subequations}
\hrulefill
\end{figure*}

We take a moment to discuss Assumption \ref{assm:n_samples}. From \eqref{eq:f_lim}, we see that \eqref{eq:target_lambda_3} is lower bounded by 
\begin{equation}
    \frac{4N_{T_k}}{9 N_s^{\ast}} \leq {\textstyle \sum_{k=1}^{N}\sum_{i=1}^{N_{T_k}} }  f(\lambda_{ik}) 
\end{equation}
In theory, this means the number of samples need to be 
\begin{equation}
    N_s \geq \frac{4 \sum_{k=1}^{N} N_{T_k}}{9\alpha}-1
\end{equation}
such that there may exist a solution that satisfies \eqref{eq:target_lambda_3}. However, since \eqref{eq:f_lim} is an asymptotic bound, more samples will be required to allow for finite values of $\lambda_{ik}$. In practice, the minimum number of samples needed will be dependent on $\alpha$, the volume of the polytopic region, number of hyperplane constraints, and the magnitude of the variance term.

We formally define the optimization problem that results from this reformulation.
\begin{subequations}\label{eq:reform_1_opt}
    \begin{align}
        \underset{\vec{U}, \lambda_{ik}}{\mathrm{minimize}} \quad & J\left(
        \bvec{X}, \vec{U}\right)  \\
        \mathrm{subject\ to} \quad  & \vec{U} \in \mathcal{U}^N,  \\
        & \text{Sample mean and sample variance-}\\ & \text{covariance matrix of } \bvec{W} \text{ defined by } \eqref{eq:stats_for_w} \nonumber \\
        & \text{Constraint \eqref{eq:target_constraint_3}} 
    \end{align}
\end{subequations}
where $\bvec{X}$ is the concatenated state vector. 

\begin{reform} \label{reform:1}
    Under Assumptions \ref{assm:samples}-\ref{assm:samples_not_equal}, solve the stochastic optimization problem \eqref{eq:reform_1_opt} with probabilistic violation threshold $\alpha$ for an open loop controller $\vec{U} \in  \mathcal U^N$ and optimization parameters $\lambda_{ik}$.
\end{reform} 

\begin{lem} \label{lem:concerv}
Solutions to Reformulation \ref{reform:1} are conservative solutions to Problem \ref{prob:1}.
\end{lem}
\begin{proof}
Lemma \ref{lem:target_satisfy} guarantees the probabilistic constraint \eqref{eq:constraint_t} is satisfied. Theorem \ref{thm:out_sample_vp} is asymptotically convergent in $N_s$ to the one-sided Vysochanskij–Petunin inequality \cite{Mercadier2021}. Since the one-sided Vysochanskij–Petunin inequality is always conservative for the Gaussian distribution, so is Theorem \ref{thm:out_sample_vp}. The sample mean and sample standard deviation terms in Reformulation \ref{reform:1} encompass and replace the dynamics used in Problem \ref{prob:1}. The cost and input constraints remain unchanged. 
\end{proof}

Here, Reformulation \ref{reform:1} is a convex optimization problem by Lemma \ref{lem:target_satisfy} and can readily be solved with off-the-shelf convex solvers.

\section{Results} \label{sec:results}

\begin{figure}
    \centering
    \includegraphics[width=0.7\columnwidth]{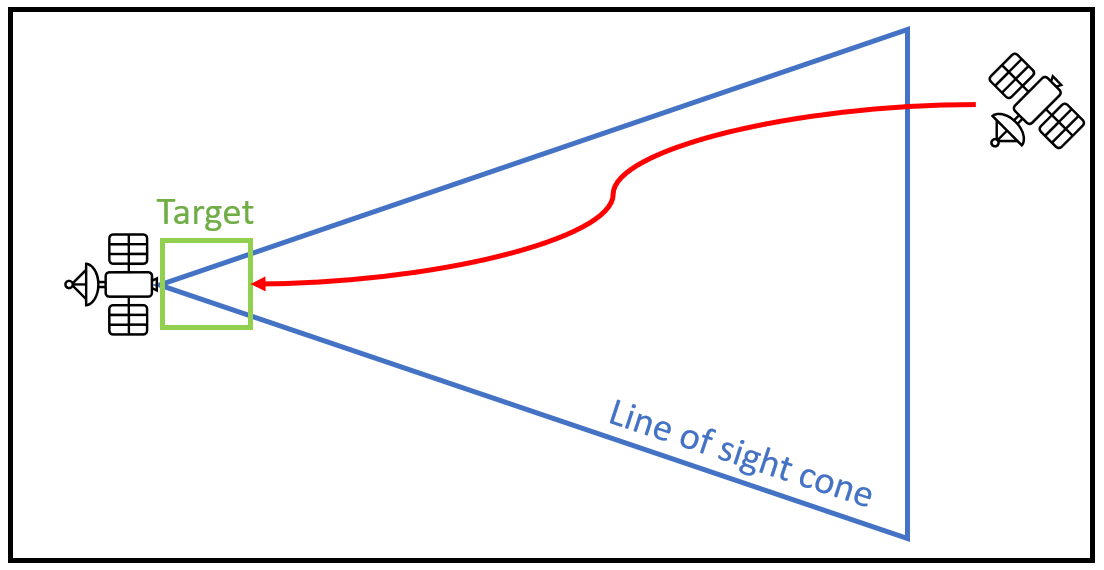}
    \caption{Graphic representation of the problem posed in Section \ref{sec:results}. Here, the dynamics of the deputy is perturbed by additive Gaussian noise with unknown mean and variance. We attempt to find a control sequence that allows the deputy to rendezvous with the chief while meeting probabilistic time varying target set requirements.}
    \label{fig:problem}
\end{figure}

We demonstrate our method on a satellite rendezvous and docking problem with simulated disturbance data. All computations were done on a 1.80GHz i7 processor with 16GB of RAM, using MATLAB, CVX \cite{cvx} and Gurobi \cite{gurobi}. All code is available at \url{https://github.com/unm-hscl/shawnpriore-sample-gaussian}. 

We consider the rendezvous of two satellites, called the deputy and chief. The deputy spacecraft must remain in a predefined line-of-sight cone, and reach a target set that describes docking at the final time step. The relative dynamics are modeled via the Clohessy–Wiltshire equations \cite{wiesel1989_spaceflight}
\begin{subequations}
\begin{align}
\ddot x - 3 \omega^2 x - 2 \omega \dot y &= F_x / m_c \label{eq:cwh:a}\\
\ddot y + 2 \omega \dot x & = F_y / m_c \label{eq:cwh:b}\\
\ddot z + \omega^2 z & = F_z/m_c. \label{eq:cwh:c}
\end{align}   
\label{eq:cwh}
\end{subequations}
with input $\vec{u} = [ \begin{array}{ccc} F_x & F_y & F_z\end{array}]^\top$, orbital rate $\omega = \sqrt{\frac{\mu}{R^3_0}}$, gravitational constant $\mu$, orbital radius $R_0$km, and mass of the deputy $m_c$ \cite{wiesel1989_spaceflight}. We discretize \eqref{eq:cwh} under the assumption of impulse control with sampling time $60$s and insert a disturbance to capture model inaccuracies such that dynamics of the deputy are described by  
\begin{equation}
    \bvec{x}(k+1) = A \bvec{x}(k) + B \vec{u}(k) + \bvec{w}(k) \label{eq:cwh_lin}
\end{equation}
with admissible input set $\mathcal{U} = [-1,1]^3$, and time horizon $N=5$, corresponding to 5 minutes of operation. 

The line-of-sight cone for time steps 1-4 is defined by
\begin{equation}
    G_k = \begin{bmatrix}
        -1 & 0 & 2 & 0 & 0 & 0 \\
        -1 & 2 & 0 & 0 & 0 & 0 \\
        -1 & 0 & -2 & 0 & 0 & 0 \\
        -1 & -2 & 0 & 0 & 0 & 0 \\
        1 & 0 & 0 & 0 & 0 & 0 
    \end{bmatrix} \; 
    \vec{h}_k = \begin{bmatrix}
        0 \\ 0 \\ 0 \\ 0 \\ 10
    \end{bmatrix}
\end{equation}
The terminal set is defined by 
\begin{equation}
    G_N = I_6 \otimes \begin{bmatrix}
        1 \\ -1
    \end{bmatrix} \; 
    \vec{h}_N = \begin{bmatrix}
        2 & 0 &  \vec{1}_{4}^{\top} & 0.1 \cdot \vec{1}_{6}^{\top}
    \end{bmatrix}^{\top}
\end{equation}
We graphically represent the problem of interest in Figure \ref{fig:problem}. The probabilistic violation threshold $\alpha$ is set to 0.05. The performance objective is based on fuel consumption, $J\left(\bvec{X}, \vec{U}\right) = \vec{U}^\top \vec{U}$.

To generate disturbance data, we randomly simulate disturbance vectors from a multivariate Gaussian with parameters
\begin{equation} \label{eq:ex_covar}
\begin{split}
    \ex{\bvec{W}} = & \; \vec{0}_{nN \times 1} \\
    \var{\bvec{W}} = & \; I_N \otimes \begin{bmatrix} 10^{-6} \cdot I_3 & 0_3 \\ 0_3 & 5 \times 10^{-8} \cdot I_3 \end{bmatrix}    
\end{split}
\end{equation}
Note that these values were only used to simulate the disturbance and were not used to evaluate the reformulated constraints \eqref{eq:target_constraint_3}.

\subsection{Comparison With Sample Based Approaches}

For this demonstration, we compare the proposed method with the scenario approach \cite{calafiore2006scenario, Campi2018TAC} and the particle control approach \cite{blackmore2010_particle}, the two most commonly used data-driven approaches. In this comparison, we expect the proposed method to be the most conservative of the three methods, leading to a higher solution cost. However, by reducing each chance constraint to an affine combination of sample statistics, we expect to see significant decrease in computation time in comparison to the other two methods. 

We compare all three methods with the same sample set. Because the scenario approach has the largest sample size requirement, we will use its sample size for all three methods. To determine the number of samples needed for the scenario approach, we use the formula
\begin{equation}
    N_s \geq \frac{2}{\alpha}\left( \ln{\frac{1}{\beta}} + N_o \right)
\end{equation}
where $\beta \in (0,1)$ is the confidence bound and $N_o$ is the number of optimization variables \cite{Campi2008}. We set $\beta = 10^{-8}$ and observe that $N_o = 15$, hence, we use $N_s = 1,337$ samples. 

Figure \ref{fig:compare_method} shows the resulting trajectories of the three methods. We see that the trajectories for the scenario approach and the particle control approach are nearly identical while the proposed method deviates from this trajectory. Here, the proposed method resulted in a trajectory that tended to have larger distances from each of the hyperplane boundaries in comparison the scenario approach solution. This is expected behavior and the embodiment of the conservatism present in our approach.

Solution statistics and empirical chance constraint satisfaction can be found in Table \ref{tab:compare_method}. We see the solution cost was larger for the proposed method, in line with our expectations. To assess constraint satisfaction, we generated $10^5$ additional disturbances and empirically tested whether the constraint was met. We expect the proposed method to always empirically satisfy the chance constraint given Lemma \ref{lem:target_satisfy}. However, this guarantee cannot be made for the scenario approach and the particle control approach. This is apparent in Table \ref{tab:compare_method} as the particle control approach did not empirically satisfy the constraint.

We point out in Table \ref{tab:compare_method} the large difference in computation time between the three methods. Note that the particle control approach was not able to find an optimal solution within a 30 minute time limit. However, despite this fact, the time to solve the proposed solution is \textit{two orders of magnitude faster than the scenario approach and at least four orders of magnitude faster than the particle control approach}. 

\begin{figure}
    \centering
    \includegraphics[width=0.8\columnwidth]{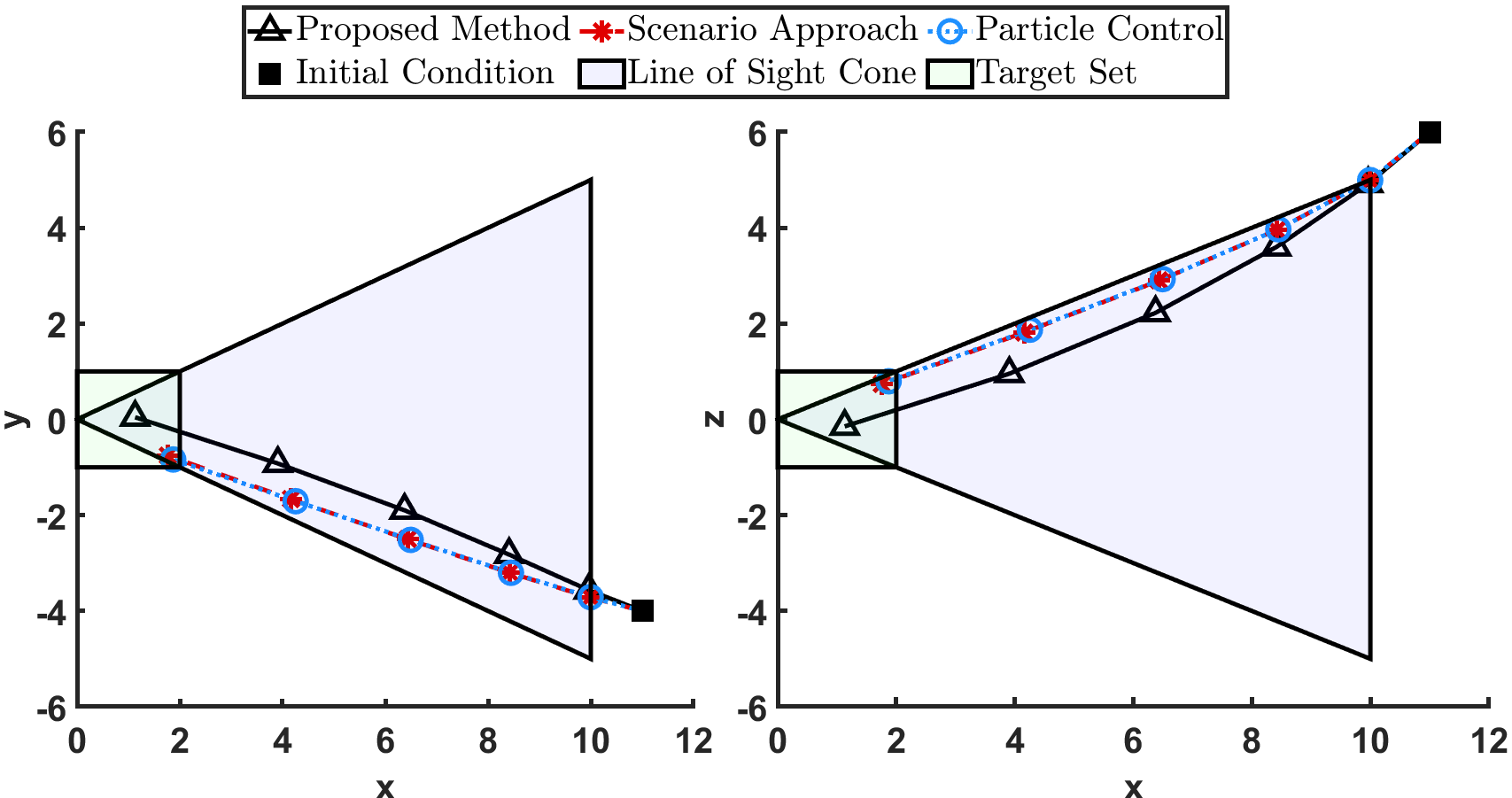}
    \caption{Comparison of mean trajectories between proposed method, scenario approach, and particle control for CWH dynamics. Note that the scenario approach and particle control trajectories are near identical.}
    \label{fig:compare_method}
\end{figure}

\begin{table*}
    \caption{Comparison of computation time, solution cost, and constraint satisfaction between proposed method, scenario approach, and particle control for CWH Dynamics with $\alpha = 0.05$. Chance constraint satisfaction is the ratio of $10^5$ samples satisfying the constraint. * indicates the method could not find an optimal solution within a 30 minute time limit.}
    \centering
    \begin{tabular}{lccc}
        \hline \hline
        Metric & Proposed & Scenario Approach \cite{calafiore2006scenario, Campi2008} & Particle Control \cite{blackmore2010_particle} \\ \hline
        Solve Time (sec) & 0.2569 & 12.2240 & 1800.0000* \\
        Cost ($N^2 \!\times\! 10^{-4}$) & $9.6118$ & $7.7886$ & $7.6691$ \\
        Constraint Satisfaction & 1.0000 & 0.9981 & 0.9432 \\
        \hline
    \end{tabular}
    \label{tab:compare_method}
\end{table*}

\subsection{Comparison With Analytic Counterpart}

We compare the proposed method with \cite{Priore2023_TAC_VP}, a chance constrained stochastic optimal control reformulation based on the one-sided Vysochanskij–Petunin inequality (shortened to OSVPI, as needed) \cite{Mercadier2021}. This approach is effective for and has been demonstrated on systems which have target set constraints represented by unimodal distributions, as is the case with the Gaussian distribution, and can be solved via convex optimization. 

\begin{table}
    \caption{Comparison of computation time, solution cost, and constraint satisfaction between proposed method and an MPC approach using the one-sided Vysochanskij-Petunin inequality (MPC/OSVPI) \cite{Priore2023_TAC_VP} for CWH dynamics with $\alpha = 0.05$. Chance constraint satisfaction is the ratio of $10^5$ samples satisfying the constraint.}
    \centering
    \begin{tabular}{lcc}
        \hline \hline
        Metric                 & Proposed   & MPC/OSVPI \cite{Priore2023_TAC_VP} \\ \hline
        Solve Time (sec)       & 0.2422         &  0.2675  \\ 
        Cost ($N^2\!\times\! 10^{-4}$)  & $ 8.3522 $    & $8.1364 $   \\
        Constraint Satisfaction & 1.0000 & 1.0000  \\
        \hline
    \end{tabular}
    \label{tab:compare_anal}
\end{table}

\begin{figure}
    \centering
    \includegraphics[width=0.8\columnwidth]{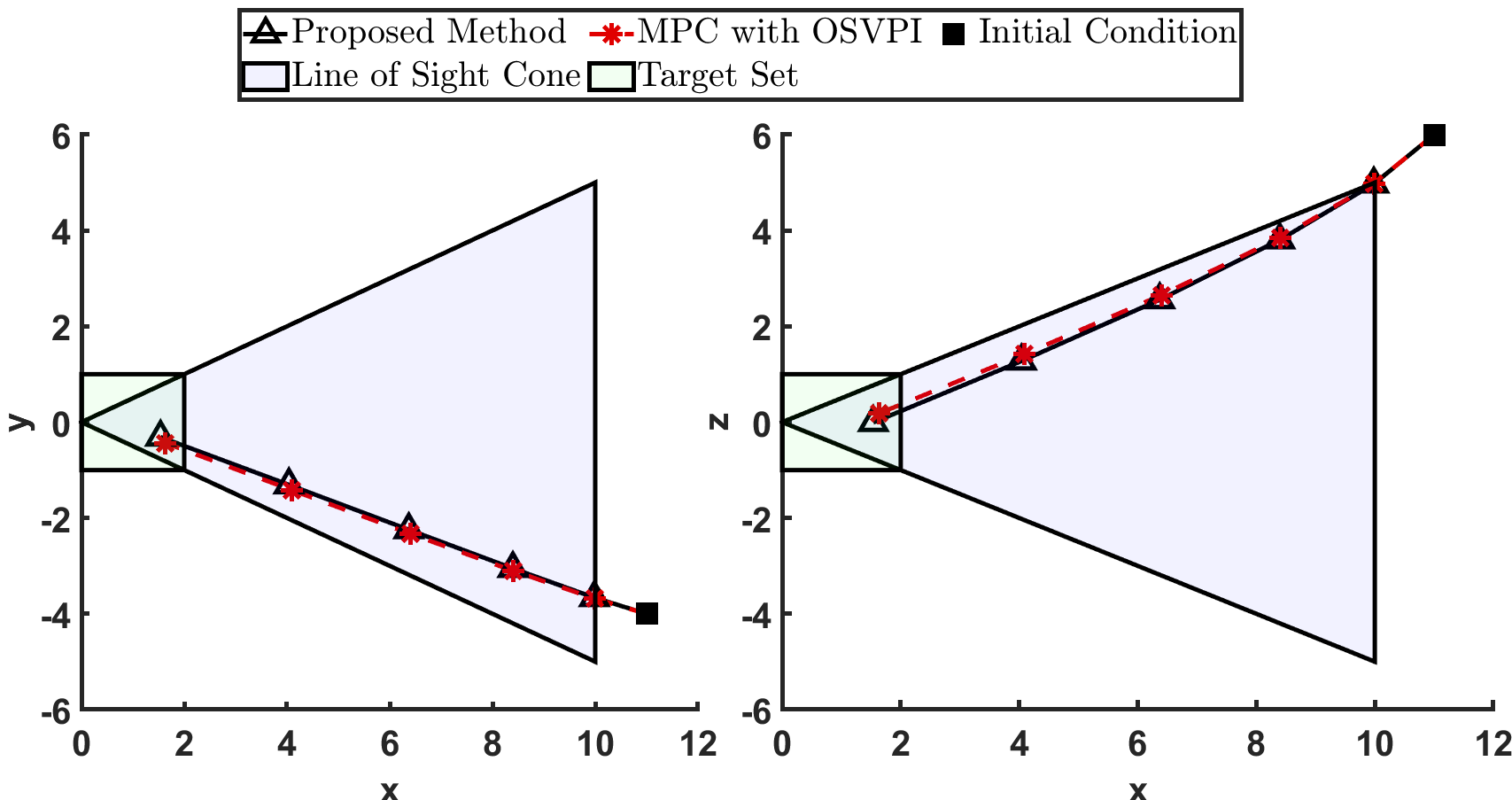}
    \caption{Comparison of trajectories between proposed method and an MPC approach using the one-sided Vysochanskij-Petunin inequality \cite{Priore2023_TAC_VP} for CWH dynamics. Note that the two trajectories are nearly identical.}
    \label{fig:compare_anal}
\end{figure}

\begin{thm}[One-sided Vysochanskij–Petunin Inequality \cite{Mercadier2021}] \label{thm:osvpi}
Let $\boldsymbol{x}$ be a real valued unimodal random variable with finite expectation $\ex{\boldsymbol{x}}$ and finite, non-zero standard deviation $\std{\boldsymbol{x}}$. Then, for $\lambda > \sqrt{5/3}$, $\pr{\boldsymbol{x} - \ex{\boldsymbol{x}}  \geq  \lambda \std{\boldsymbol{x}}} \leq \frac{4}{9(\lambda^2+1)}$.
\end{thm}

As mentioned in the proof of Lemma \ref{lem:concerv}, Theorem \ref{thm:out_sample_vp} is asymptotically convergent in $N_s$ to the one-sided Vysochanskij–Petunin inequality. In this demonstration, we show that Theorem \ref{thm:out_sample_vp} does not add significant conservatism in comparison to the one-sided Vysochanskij–Petunin inequality for large but finite sample sizes. Here, we have selected the sample size for the proposed method to be $N_s = 5,000$. For comparison, we use the mean and covariance matrix \eqref{eq:ex_covar} to compute a solution with the method of \cite{Priore2023_TAC_VP}.

Figure \ref{fig:compare_anal} shows the resulting trajectories and Table \ref{tab:compare_anal} compares the time to compute a solution, solution cost, and empirical chance constraint satisfaction with $10^5$ additional samples disturbances. The two methods preformed near identically. We see this in the resulting trajectories and computation time. The only notable difference is that the proposed method resulted in a 2.6\% increase in solution cost. This is a small increase if we consider the proposed method does not require full knowledge of the underlying distribution. Here, we've shown that the proposed method results in only a small deviation for a finite sample size in comparison to that of its asymptotic counterpart as in \cite{Priore2023_TAC_VP}.

\section{Conclusion} \label{sec:conclusion}

We proposed a sample-based framework to solve chance-constrained stochastic optimal control problems with probabilistic guarantees. This work focuses on joint chance constraints for polytopic target sets in Gaussian disturbed LTI systems where the disturbance’s mean and variance are unknown. We derived a concentration inequality that allow us to bound tail probabilities of a random variable being a set number of sample standard deviations away from the sample mean under unimodality conditions. Our approach relies on this derived theorem to reformulate joint chance constraints into a series of inequalities that can be readily solved as a convex optimization problem. We demonstrated our method on a multi-satellite rendezvous problem and compared it with the scenario approach, particle control, and an MPC method based on the one-sided Vysochanskij-Petunin inequality. We showed the proposed method can find a solution in significantly less time than comparable sample based approaches, and results in only a small increase in conservatism despite using sample statistics. 


\appendix

\subsection{Proof of Theorem \ref{thm:out_sample_vp}} \label{appx:out-sample}

To prove Theorem \ref{thm:out_sample_vp}, we first need to state and prove the following Lemma. 

\begin{lem} \label{thm:in_sample_vp}
Let $\boldsymbol{x}^{[1]},\dots, \boldsymbol{x}^{[N_s]}$ be samples drawn i.i.d., for some $N_s\geq2$. Let
\begin{subequations}
\begin{align}
    \hex{\boldsymbol{x}} = & \; \frac{1}{N_s} \sum_{i=1}^{N_s} \boldsymbol{x}^{[i]} \\
    \hstd{\boldsymbol{x}} =  & \; \sqrt{\frac{1}{N_s} \sum_{i=1}^{N_s} (\boldsymbol{x}^{[i]} - \hex{\boldsymbol{x}})^2}
\end{align}
\end{subequations}
be the sample mean and sample standard deviation, respectively, with $\hstd{\boldsymbol{x}}>0$ almost surely. Then, if the distribution of the statistic $\frac{\boldsymbol{x}^{[i]}-\hex{\boldsymbol{x}}}{ \hstd{\boldsymbol{x}}}$ is unimodal,
$\pr{\boldsymbol{x}^{[i]}-\hex{\boldsymbol{x}} \geq \lambda \hstd{\boldsymbol{x}}} \leq \frac{4}{9(\lambda^2+1)}$ for any $\lambda>\sqrt{5/3}$ and $i \in \Nt{1}{N_s}$.
\end{lem}

\begin{proof}
For brevity we denote
\begin{equation}
    \hat{\mathbb{G}}(\boldsymbol{x}) = \frac{\boldsymbol{x}^{[i]}-\hex{\boldsymbol{x}}}{\hstd{\boldsymbol{x}}}
\end{equation}
Without loss of generality, for $\delta > 0$,
\begin{equation}\label{eq:proof_1}
    \pr{\boldsymbol{x}^{[i]}\!-\!\hex{\boldsymbol{x}} \!\geq\! \lambda \hstd{\boldsymbol{x}}}
    \! \leq\! \pr{\left| \hat{\mathbb{G}}(\boldsymbol{x}) \!+\!\delta \right| \!\geq\! \lambda \!+\! \delta} 
\end{equation}
Here, $\hat{\mathbb{G}}(\boldsymbol{x})$ is unimodal, $\E[\hat{\mathbb{G}}(\boldsymbol{x})]=0$, and $\E[\hat{\mathbb{G}}(\boldsymbol{x})^2]=\mathrm{std}(\hat{\mathbb{G}}(\boldsymbol{x}))=1$. So, by the Vysochanskij-Petunin inequality \cite{Vysochanskij1980} for $\left(\lambda + \delta \right)^2 \geq \frac{8}{3} \E[(\hat{\mathbb{G}}(\boldsymbol{x})+\delta)^2]$, we bound

\begin{subequations}\label{eq:simpex}
\begin{align}
     \pr{\left| \hat{\mathbb{G}}(\boldsymbol{x}) \!+\!\delta \right| \!\geq\! \lambda \!+\! \delta}  & \ \leq  \frac{4}{9} \frac{\E[(\hat{\mathbb{G}}(\boldsymbol{x})+\delta)^2]}{ \left(\lambda + \delta \right)^2}\\
    & \ =  \frac{4}{9} \frac{1 + \delta^2}{ \left(\lambda + \delta \right)^2} \\
    & \ =  \frac{4}{9(\lambda^2  + 1)}
\end{align}
\end{subequations}
as $\lambda^{-1}$ is the optimal value of $\delta$. 

Finally, using the same logic as in \eqref{eq:simpex}, we simplify the condition on $\lambda$ 
\begin{subequations}
\begin{align}
    \left(\lambda + \delta \right)^2 \geq & \; \frac{8}{3} \E[(\hat{\mathbb{G}}(\boldsymbol{x})+\delta)^2]\\
    \left(\lambda + \lambda^{-1} \right)^2 \geq & \;  \frac{8}{3} \left(1+\lambda^{-2}\right) \\
    \lambda \geq & \;  \sqrt{\frac{5}{3}}
\end{align}
\end{subequations}
\end{proof}

With Lemma \ref{thm:in_sample_vp} established, we can prove Theorem \ref{thm:out_sample_vp}.

\begin{proof}[Proof of Theorem \ref{thm:out_sample_vp}]
Let $\hex{\boldsymbol{x}}^{\ast}$ and $\hvar{\boldsymbol{x}}^{\ast}$ denote the sample mean and variance computed with $N_s + 1$ samples. Note that
\begin{subequations}
\begin{align}
    \boldsymbol{x}^{[N_s^{\ast}]} \!-\!\hex{\boldsymbol{x}}^{\ast} & = \; \frac{N_s}{N_s^{\ast}}(\boldsymbol{x}^{[N_s^{\ast}]}-\hex{\boldsymbol{x}}) \\
    \hvar{\boldsymbol{x}}^{\ast}  & = \; \frac{N_s}{N_s^{\ast}} \hvar{\boldsymbol{x}} \!+\! \frac{N_s}{N_s^{\ast2}} (\boldsymbol{x}^{[N_s^{\ast}]}-\hex{\boldsymbol{x}})^2 \label{eq:sample_var_conversion}
\end{align}
\end{subequations}
Then for $\lambda >0$
\begin{subequations}
\begin{align}
    & \pr{\boldsymbol{x}^{[N_s^{\ast}]}-\hex{\boldsymbol{x}} \geq \lambda \hstd{\boldsymbol{x}}} \\
    & \; = \P\left((\sqrt{N_s N_s^{\ast}}+\lambda\sqrt{N_s})(\boldsymbol{x}^{[N_s^{\ast}]}-\hex{\boldsymbol{x}}) \right.\\ 
    & \qquad \quad \left. \geq \lambda \sqrt{N_s N_s^{\ast}}\hstd{\boldsymbol{x}} +  \lambda\sqrt{N_s}(\boldsymbol{x}^{[N_s^{\ast}]}-\hex{\boldsymbol{x}})\right) \nonumber  \\
    & \; \leq \P\left((\sqrt{N_s N_s^{\ast}}+\lambda\sqrt{N_s})(\boldsymbol{x}^{[N_s^{\ast}]}-\hex{\boldsymbol{x}})\right. \label{eq:out_of_sample_triangle}\\
    & \qquad \quad \left.\geq \lambda \sqrt{N_s N_s^{\ast}\hvar{\boldsymbol{x}} + N_s (\boldsymbol{x}^{[N_s^{\ast}]}-\hex{\boldsymbol{x}})^2}\right) \nonumber
\end{align}
where \eqref{eq:out_of_sample_triangle} results from the triangle inequality. Then,
\begin{align}    
    &\P\left((\sqrt{N_s N_s^{\ast}}+\lambda\sqrt{N_s})(\boldsymbol{x}^{[N_s^{\ast}]}-\hex{\boldsymbol{x}})\right.\\
    & \qquad \quad \left.\geq \lambda \sqrt{N_s N_s^{\ast}\hvar{\boldsymbol{x}} + N_s (\boldsymbol{x}^{[N_s^{\ast}]}-\hex{\boldsymbol{x}})^2}\right) \nonumber\\
    & \; = \pr{\boldsymbol{x}^{[N_s^{\ast}]} \!-\!\hex{\boldsymbol{x}} \geq \frac{\lambda N_s^{\ast}}{ \sqrt{N_s N_s^{\ast}}+\lambda\sqrt{N_s}}\hstd{\boldsymbol{x}}^{\ast} } \\
    & \; = \pr{\boldsymbol{x}^{[N_s^{\ast}]}-\hex{\boldsymbol{x}}^{\ast} \geq \frac{\lambda \sqrt{N_s}}{\sqrt{N_s^{\ast}}+\lambda}\hstd{\boldsymbol{x}}^{\ast} } \\
    & \; = \pr{\boldsymbol{x}^{[N_s^{\ast}]}-\hex{\boldsymbol{x}}^{\ast} \geq \kappa \hstd{\boldsymbol{x}}^{\ast}} \label{eq:proof_2_1}
\end{align}
Where $\kappa$ is a simple substitution. Here, $\lambda >\frac{\sqrt{5N_s^{\ast}}}{\sqrt{3N_s}-\sqrt{5}}$ implies $\kappa>\sqrt{5/3}$. So, by Lemma \ref{thm:in_sample_vp},

\begin{align}
    \pr{\boldsymbol{x}^{[N_s^{\ast}]}\!-\!\hex{\boldsymbol{x}}^{\ast} \!\geq\! \kappa \hstd{\boldsymbol{x}}^{\ast}}
    & \; \leq \frac{4}{9(\kappa^2 + 1)} \\
    & \; = f(\lambda) \label{eq:out_of_sample_result}
\end{align}
\end{subequations}
\end{proof}

\bibliographystyle{ieeetr}
\bibliography{main}

\end{document}